\documentclass[12pt,a4paper]{article}
\usepackage[latin1]{inputenc}
\usepackage{amsmath}
\usepackage{amsfonts}
\usepackage{amssymb}
\usepackage{bbm}
\usepackage{tabularx}
\usepackage{graphicx}
\usepackage{amsthm}
\usepackage{mathtools}
\usepackage{fullpage}

\newtheorem{lemma}{Lemma}
\newtheorem{cor}{Corollary}
\newtheorem{theorem}{Theorem}

\author{Patrick Schnider}
\title{A generalization of crossing families}
\date{}
\begin{document}
\maketitle

\begin{abstract}
For a set of points in the plane, a \emph{crossing family} is a set of line segments, each joining two of the points, such that any two line segments cross. We investigate the following generalization of crossing families: a \emph{spoke set} is a set of lines drawn through a point set such that each unbounded region of the induced line arrangement contains at least one point of the point set. We show that every point set has a spoke set of size $\sqrt{\frac{n}{8}}$. We also characterize the matchings obtained by selecting exactly one point in each unbounded region and connecting every such point to the point in the antipodal unbounded region.
\end{abstract}

\section{Introduction}
Let $\mathcal{P}$ be a finite point set in general position (i.e.\ no three points on a line). A \emph{crossing family} is a set of line segments, each joining two of the points, such that any two line segments cross (i.e.\ intersect in their interior). Crossing families were introduced by Aronov et al.\ \cite{Aronov}, who have shown that any set of $n$ points in general position has a crossing family of size $\sqrt{\frac{n}{12}}$. Since then, there have been several results about crossing families \cite{Fulek, Pach}, but even though it is conjectured that any point set in general position has a crossing family of linear size, the bound of Aronov et al.\ is still the best known result.

A point set $\mathcal{A}$ \emph{avoids} a point set $\mathcal{B}$ if no line through two points in $\mathcal{A}$ intersects the convex hull of $\mathcal{B}$. If $\mathcal{B}$ also avoids $\mathcal{A}$, the two sets are called \emph{mutually avoiding}. The bound in \cite{Aronov} on the size of the largest crossing family is proven in two steps: first it is shown that two mutually avoiding sets $\mathcal{A}$ and $\mathcal{B}$, each of size $k$, induce a crossing family of size $k$. Then it is shown that every set of $n$ points in general position contains two mutually avoiding sets of size $\sqrt{\frac{n}{12}}$. In this paper we will follow the same approach, but for a generalization of crossing families.

Bose et al.\ \cite{Bose} have introduced the following generalization of crossing families: A \emph{spoke set of size $k$} for $\mathcal{P}$ is a set $\mathcal{S}$ of $k$ pairwise non-parallel lines such that in each unbounded region of the arrangement defined by the lines in $\mathcal{S}$ there lies at least one point of $\mathcal{P}$. Note that it is easy to obtain a spoke set from a crossing family by slightly rotating the supporting lines of the line segments in the crossing family. Then each endpoint of a line segment in the crossing family lies in a different unbounded region. We will show that every set of $n$ points in general position contains a spoke set of size $\sqrt{\frac{n}{8}}$. To this end, we first translate the notion of spoke sets to the dual setting in Section \ref{Sec:duality}. In Section \ref{Sec:construction} we then use the dual version to construct large spoke sets for the union of two point sets $\mathcal{A}$ and $\mathcal{B}$, where $\mathcal{A}$ avoids $\mathcal{B}$ and $\mathcal{A}$ and $\mathcal{B}$ can be separated by a line. Finally, we show that every point set contains such point sets $\mathcal{A}$ and $\mathcal{B}$ and give bounds on their sizes.

The motivation for the introduction of spoke sets in \cite{Bose} is the fact that with a spoke set of size $k$ for $\mathcal{P}$, one can construct a covering of the edge set of the complete geometric graph drawn on $\mathcal{P}$ with $n-k$ crossing-free spanning trees. The result in this paper thus also improves the previous upper bound of $n-\sqrt{\frac{n}{12}}$ for this problem. However, the original question from \cite{Bose}, whether there is always a spoke set of linear size, remains open.

Another interesting question is whether it is always possible to find a crossing family of size linear in the size of the largest spoke set. Theorem \ref{Thm:charac} is a first step in this direction as it characterizes the matchings obtained from spoke sets and shows that even though they might not all be crossing families, they still satisfy a number of conditions.

For space reasons, we will not be able to give all proofs. Instead, we refer the interested reader to the appendix.

\subsection*{Preliminaries}

Let $\mathcal{S}$ be a spoke set of size $k$ for $\mathcal{P}$. Consider the ordering of $\mathcal{S}=\{\ell_1,\ldots,\ell_k\}$ by increasing slope. Let $U^+_i$ be the unbounded region that lies below $\ell_1,\ldots, \ell_i$ and above $\ell_{i+1},\ldots \ell_k$. Similarly let $U^-_i$ be the unbounded region that lies above $\ell_1,\ldots, \ell_i$ and below $\ell_{i+1},\ldots \ell_k$. We call the regions $U^+_i$ and $U^-_i$ \emph{antipodal}.

Let $\mathcal{Q}$ be a subset of $\mathcal{P}$ that has exactly one point in each unbounded region. Note that then each line of $\mathcal{S}$ is a halving line for $\mathcal{Q}$. The \emph{spoke matching} of $\mathcal{Q}$ is the matching obtained by drawing a straight line segment from each point $p$ in $\mathcal{Q}$ to the unique point $q$ in $\mathcal{Q}$ that lies in the antipodal unbounded region of the spoke set. See Figure \ref{Fig:SpokeSetDefinition} for an example. Note that in a spoke matching, each edge intersects every line of the spoke set. In Section \ref{Sec:matchings} we characterize the geometric matchings that are spoke matchings.

\begin{figure}
\centering
\includegraphics[scale=0.7]{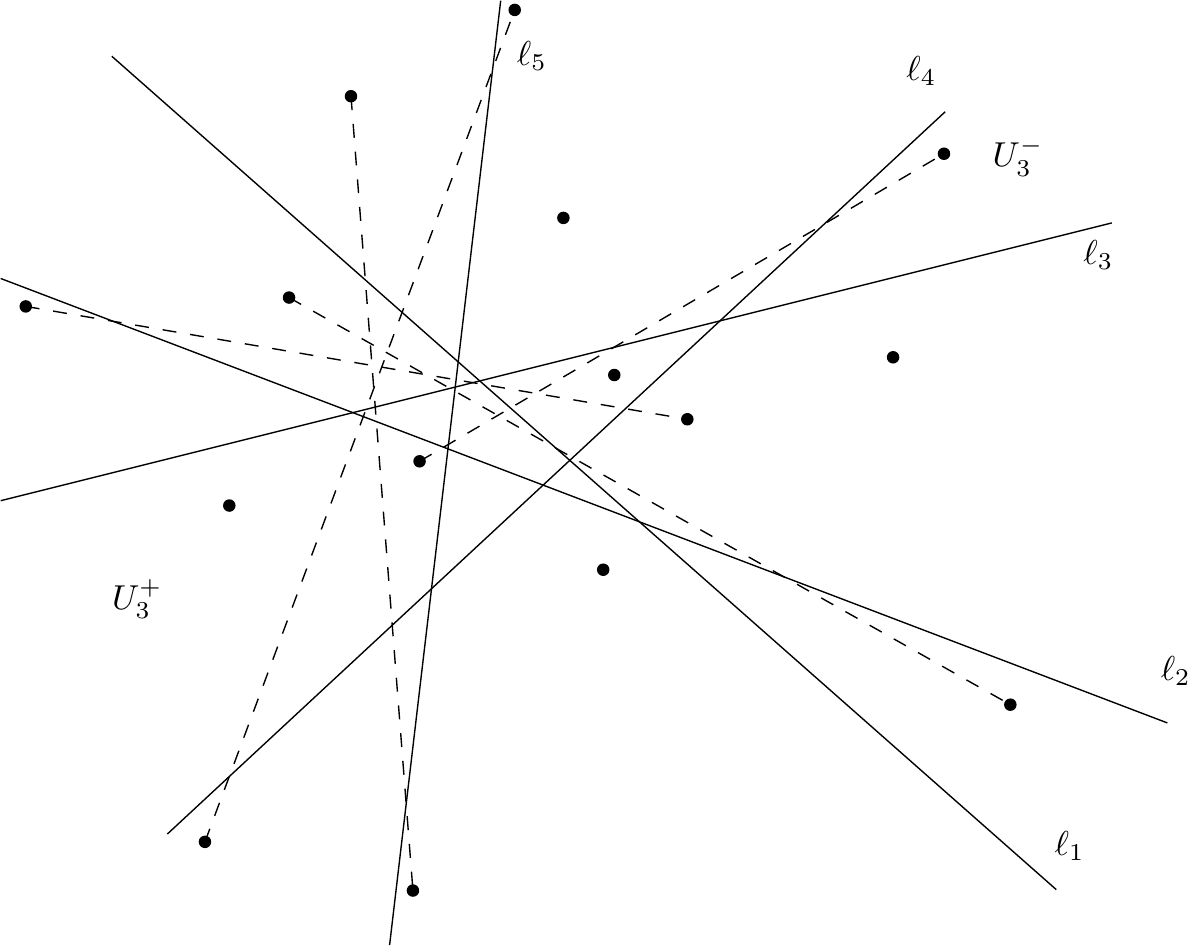}
\caption{A spoke set and a spoke matching (dashed).}
\label{Fig:SpokeSetDefinition}
\end{figure}

\section{Spoke sets under duality}
\label{Sec:duality}
In this section we will translate the properties of spoke sets into the dual setting, that is under the point-line duality. For this we start with some definitions.

Given an arrangement $\mathcal{A}$ of lines, without loss of generality none of them horizontal or vertical, a \emph{cell-path} $R$ is a sequence of cells such that consecutive cells share an edge. If the edge shared by two consecutive cells is a subset of some line $a_i$ of $\mathcal{A}$, we say that $R$ \emph{crosses} $a_i$. The \emph{length} of a cell-path is one less than the number of cells involved. We call a cell-path \emph{line-monotone} if it crosses each line of $\mathcal{A}$ at most once.

If $\mathcal{A'}$ is an arrangement induced by a subset of the lines of $\mathcal{A}$, then $R$ \emph{restricted to} $\mathcal{A'}$ is the cell path obtained by replacing each cell $C$ of $\mathcal{A}$ in $R$ by the cell $C'$ in $\mathcal{A'}$ with $C\subseteq C'$ and deleting consecutive multiples.

Finally, for a cell-path $R=(C_0,C_1,\ldots,C_k)$, let $a_i$ be the line in $\mathcal{A}$ that contains the edge shared by $C_{i}$ and $C_{i+1}$. We call the pair $(a_{2j},a_{2j+1})$ \emph{AB-alternating}, if $C_{2j}$ lies below $a_{2j}$ and $C_{2j+1}$ lies above $a_{2j+1}$ or the other way around, that is $C_{2j}$ lies above $a_{2j}$ and $C_{2j+1}$ lies below $a_{2j+1}$. We call a cell path $P=(C_0,C_1,\ldots,C_{2k})$ \emph{AB-semialternating} if for every $j<k$ the pair $(a_{2j},a_{2j+1})$ is AB-alternating. See Figure \ref{Fig:SpokePathsDefinition} for an example.

\begin{figure}
\centering
\includegraphics[scale=0.7]{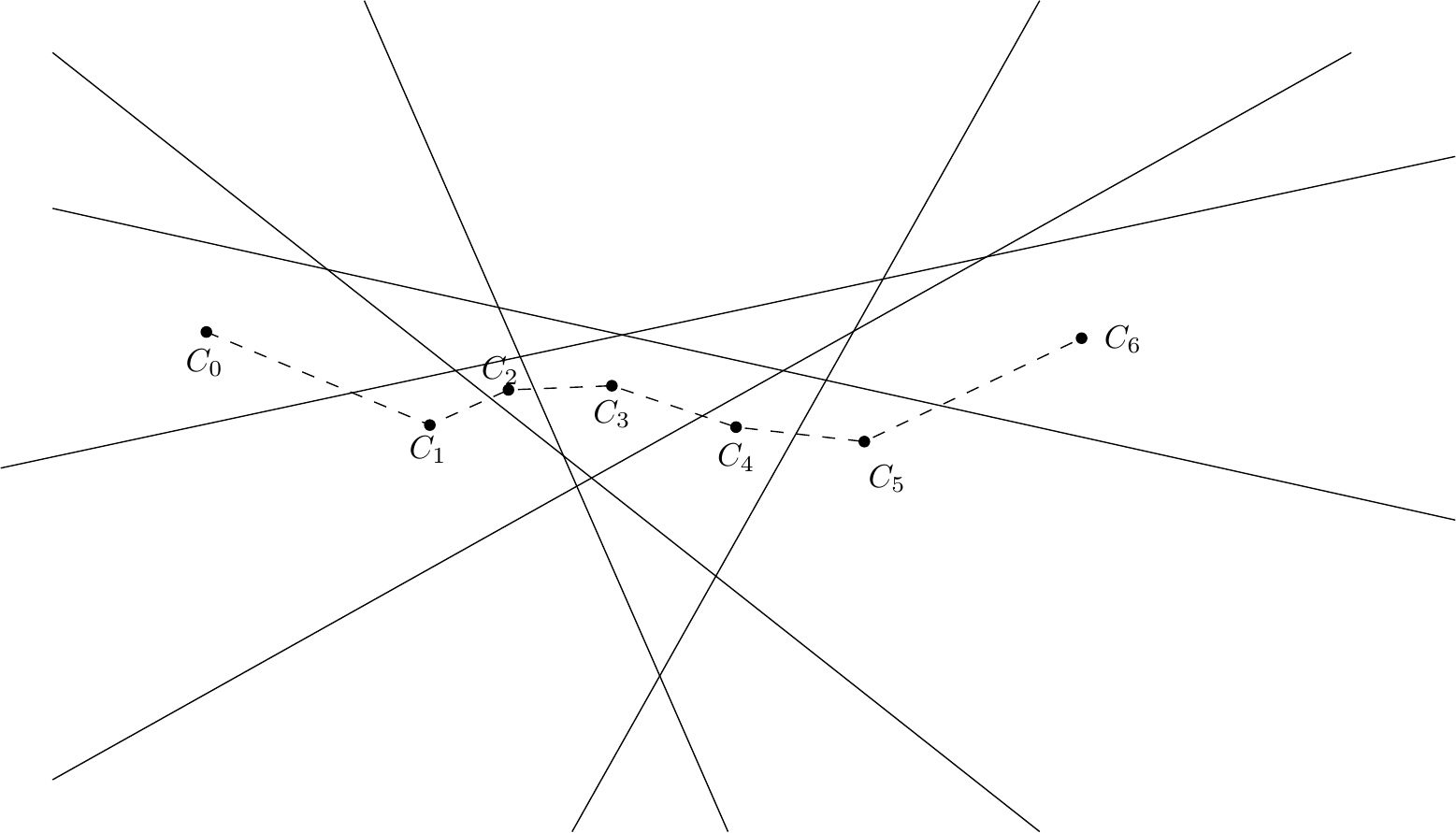}
\caption{A line-monotone AB-semialternating cell-path of length 6.}
\label{Fig:SpokePathsDefinition}
\end{figure}

We now have all the vocabulary that is necessary to describe the dual of spoke sets: given an arrangement $\mathcal{A}$ of lines, a \emph{spoke path} $(R,\mathcal{A'})$ is a cell-path $R$ together with an arrangement $\mathcal{A'}$ induced by a subset of the lines of $\mathcal{A}$, such that $R$ restricted to $\mathcal{A'}$ is line-monotone and AB-semialternating. The length of a spoke path $(R,\mathcal{A'})$ is the length of $R$ restricted to $\mathcal{A'}$.
Note that all the definitions generalize to $x$-monotone pseudoline arrangements.

\begin{lemma}
\label{Lem:dual}
Let $\mathcal{P}$ be a point set and $\mathcal{P^*}$ its dual line arrangement. Then $\mathcal{P}$ contains a spoke set of size $k$ if and only if $\mathcal{P^*}$ contains a spoke path of length $2k$.
\end{lemma}

\begin{proof}
First assume that $\mathcal{P^*}$ contains a spoke path $(R,\mathcal{Q^*})$ of length $2k$. Let $\mathcal{Q}$ be the primal of $\mathcal{Q^*}$. We want to show that $\mathcal{Q}$, and thus $\mathcal{P}$, contains a spoke set of size $k$. Let $R$ restricted to $\mathcal{Q^*}$ be $(C_0,C_1,\ldots,C_{2k})$. Let $s_i$ be a line in the primal such that its dual point lies in $C_i$. Let $\mathcal{S}=\{s_0,s_2,s_4,\ldots,s_{2j},\ldots,s_{2k-2}\}$. Our goal is to show that $\mathcal{S}$ is a spoke set.

As $R$ crosses each line of $\mathcal{Q^*}$ exactly once, half of them from below and half of them from above, the cells $C_0$ and $C_{2k}$ must be the unbounded cells of the middle level of the line arrangement $\mathcal{Q^*}$. Assume without loss of generality that $C_0$ is the left unbounded cell of the middle level. Then $s_0$ and $s_{2k}$ are lines in the primal that are halving lines for $\mathcal{Q}$ and have the same subsets of $\mathcal{Q}$ as semispaces, with $s_0$ having negative slope and $s_{2k}$ having positive slope. Also, as $R$ restricted to $\mathcal{Q^*}$ is AB-semialternating, the cells $C_0, C_2, \ldots, C_{2j}, \ldots, C_{2k}$ must all be in the middle level, with $C_{2j-2}$ being to the left of $C_{2j}$ for every $j$. In particular, the lines $s_0,s_2,s_4,\ldots,s_{2j},\ldots,s_{2k-2}$ in the primal are ordered by increasing slope.

Let $U^+_i$ be the unbounded region in the primal that lies below $s_0, s_2. \ldots, s_{2i}$ and above $s_{2i+2},\ldots s_{2j}$. Similarly let $U^-_i$ be the unbounded region that lies above $s_0, s_2. \ldots, s_{2i}$ and below $s_{2i+2},\ldots s_{2j}$. Let $a^*_i$ and $b^*_i$ be the lines in the dual that $R$ crosses between $C_{2i}$ and $C_{2i+2}$ and assume without loss of generality that $a^*_i$ is crossed from above and $b^*_i$ from below. In particular, all the cells $C_0, \ldots, C_{2i}$ lie above $a^*_i$ and below $b^*_i$, whereas all the cells $C_{2i+2},\ldots C_{2j}$ lie below $a^*_i$ and above $b^*_i$. This implies that $a_i$, the primal of $a^*_i$, lies above $s_0, s_2. \ldots, s_{2i}$ and below $s_{2i+2},\ldots s_{2j}$, i.e. in $U^-_i$. Similarly, $b_i$, the primal of $b^*_i$, lies in $U^+_i$. Thus no unbounded region is empty and $S$ is indeed a spoke set.

Now, assume that $\mathcal{P}$ contains a spoke set $s_0, \ldots, s_{k-1}$ of size $k$, ordered by increasing slope. Assume without loss of generality that $s_0$ is vertical. Slightly rotate the configuration in counter-clockwise direction, such that $s_0$ has negative slope. Let $s_k$ be a line generated by slightly rotating $s_0$ clockwise until it has positive slope. Let $\mathcal{Q}$ be a subset of $\mathcal{P}$ that has exactly one point in each unbounded region and let $\mathcal{Q^*}$ be its dual line arrangement. Let $C_{2i}$ be the cell in $\mathcal{Q^*}$ that contains the dual of $s_i$. Then the cells $C_0, C_2, \ldots, C_{2j}, \ldots, C_{2k}$ are all in the middle level, with $C_{2j-2}$ being to the left of $C_{2j}$ for every $j$. For every $i$, draw a straight edge $e_i$ between the dual of $s_i$ and the dual of $s_{i+1}$. This edge intersects exactly 2 lines $a^*_i$ and $b^*_i$ from $\mathcal{Q^*}$. Assume without loss of generality that $e_i$ does not go through the intersection of $a^*_i$ and $b^*_i$ (otherwise curve it slightly). Then $e_i$ intersects exactly 3 cells of $\mathcal{Q^*}$, namely $C_{2i}, C_{2i+2}$, and a third cell, which we denote by $C_{2i+1}$. As both $C_{2i}$ and $C_{2i+2}$ are in the middle level, $e_i$ must cross one of $a^*_i$ and $b^*_i$ from above and the other one from below. Also, the union of all $e_i$'s is monotone with respect to the $x$-axis, i.e.\ no line of $\mathcal{Q^*}$ is crossed more than once. Thus $(C_0, C_1,\ldots, C_{2k})$ is line-monotone and AB-semialternating.
\end{proof}

\section{Finding large spoke sets}
\label{Sec:construction}
In this section, we will construct large spoke sets by constructing long spoke paths in the dual arrangement. We do this for the special case of so called avoiding sets, and then show that every point set contains large avoiding subsets.

Recall that a point set $\mathcal{A}$ \emph{avoids} a point set $\mathcal{B}$ if no line through two points in $\mathcal{A}$ intersects the convex hull of $\mathcal{B}$. Let $\mathcal{A}$ and $\mathcal{B}$ be two disjoint point sets such that $\mathcal{A}$ avoids $\mathcal{B}$ and $\mathcal{A}$ and $\mathcal{B}$ can be separated by a line $\ell$. Let $\mathcal{A^*}$ be the dual arrangement of $\mathcal{A}$ and let $\mathcal{B^*}$ be the dual arrangement of $\mathcal{B}$. Let $\mathcal{P}$ be the union of $\mathcal{A}$ and $\mathcal{B}$ and let $\mathcal{P^*}$ be its dual arrangement.

Assume without loss of generality that $\ell$ is vertical and that $\mathcal{A}$ lies to the left of $\ell$ (and thus $\mathcal{B}$ to the right). Consider the dual arrangement $\mathcal{P^*}$ and let $C_0$ be the left unbounded region of the middle level. Then $C_0$ lies above all lines of $\mathcal{A^*}$ and below all lines of $\mathcal{B^*}$. Also, as $\mathcal{A}$ avoids $\mathcal{B}$, all lines of $\mathcal{B^*}$ intersect the lines of $\mathcal{A^*}$ in the same order. Draw $\mathcal{B^*}$ as a wiring diagram with color red such that all the pseudolines are between the horizontal lines $y=-\alpha$ and $y=\alpha$ for some $\alpha>0$. Then draw $\mathcal{A^*}$ with blue pseudolines that are vertical between $y=-\alpha$ and $y=\alpha$ such that the resulting pseudoline-arrangement $\mathcal{P}_0$ is isomorphic to $\mathcal{P^*}$. As all lines of $\mathcal{B^*}$ intersect the lines of $\mathcal{A^*}$ in the same order, this can always be done. We call such a drawing an \emph{extended diagram}. See Figure \ref{Fig:ExtendedDiagram} for an Example. Note that the cell corresponding to $C_0$ is under all red pseudolines and left of all blue pseudolines.

\begin{figure}
\label{Fig:ExtendedDiagram}
\centering
\includegraphics[scale=0.7]{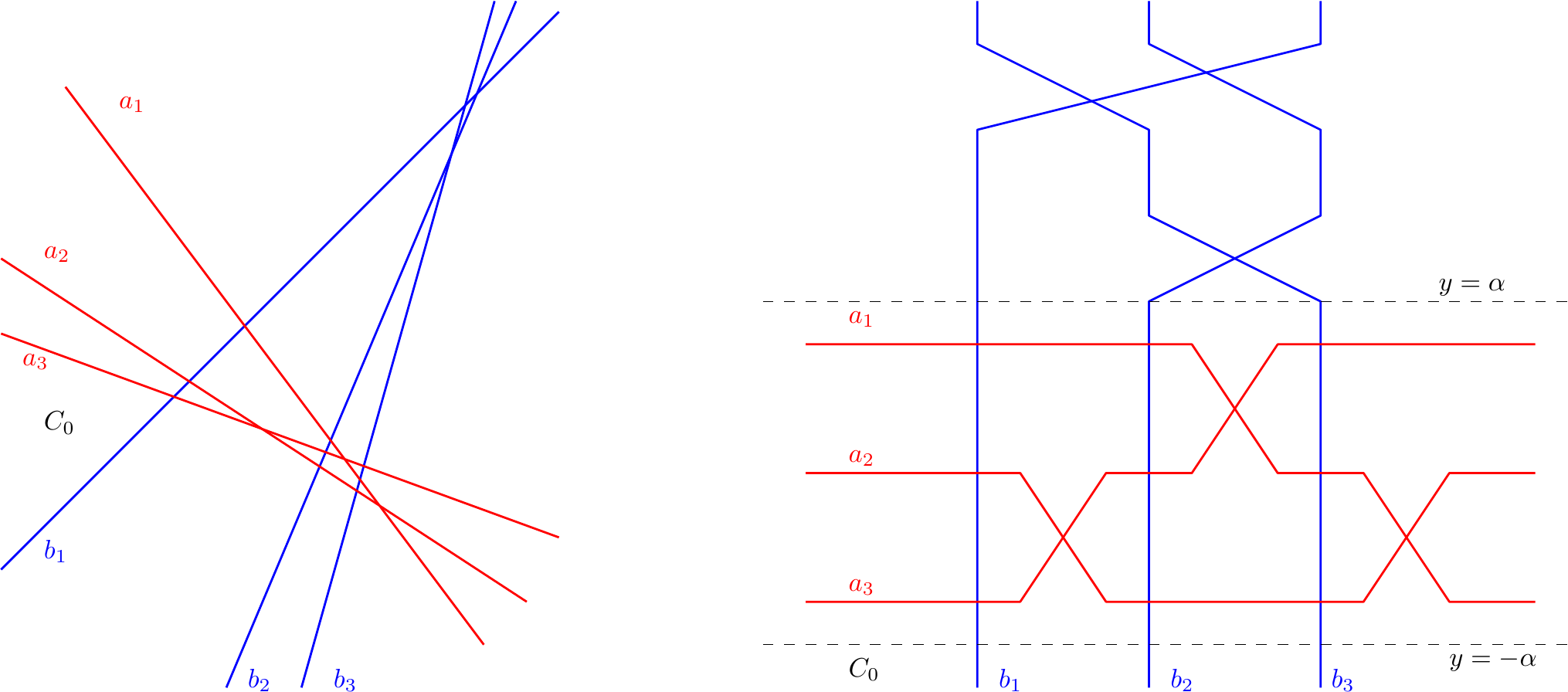}
\caption{A line arrangement and its extended diagram.}
\end{figure}

\begin{lemma}
\label{Lem:wiring}
Let $\mathcal{A}$ and $\mathcal{B}$ be two disjoint point sets of size $k$ such that $\mathcal{A}$ avoids $\mathcal{B}$ and $\mathcal{A}$ and $\mathcal{B}$ can be separated by a line. Let $\mathcal{P}=\mathcal{A}\cup\mathcal{B}$. Then the dual arrangement $\mathcal{P^*}$ contains a spoke path of length $k+2$, if $k$ is even, or $k+3$, if $k$ is odd.
\end{lemma}

\begin{proof}
Consider the extended diagram of $\mathcal{P^*}$. We will now draw a directed pseudoline $g$, representing a cell path, through this extended diagram. The pseudoline $g$ will start in the cell $C_0$ and end in the antipodal unbounded cell (that is, the cell corresponding to the right unbounded region of the middle level in $\mathcal{P^*}$). The cell-path will be given by the cells $g$ intersects. As $g$ is a pseudoline, this cell-path is certainly line-monotone. In order to get a long spoke path, we also try to draw $g$ in such a way, that within the sequence of crossed pseudolines we have many alternations between red and blue pseudolines. More formally, we define the \emph{color sequence} $c(g)$ of $g$ by moving along $g$ and writing for each crossing with another pseudoline an $r$ or a $b$ if the crossed pseudoline is red or blue, respectively. Thus, the color sequence is of the form $x_1, \ldots, x_{2k}$, with $x_i\in\{b,r\}$. We define $\phi(g)$ as the length of the longest \emph{semialternating} subsequence $x'_1,\ldots,x'_j$ of the color sequence of $g$, that is, the longest subsequence that has the following properties:
\begin{itemize}
\item $j$ is even, i.e.\ $j=2m$;
\item for every $i\leq m$ we have that $x'_{2i}=r\Leftrightarrow x'_{2i+1}=b$.
\end{itemize}
Then the cell path given by $g$, restricted to the arrangement of pseudolines that contibute to the crossings in the subsequence $x'_1,\ldots,x'_j$ is also AB-semialternating.

The way we will construct $g$ is the following: we modify the extended diagram in a sequence of steps until we reach a certain goal diagram where we can guarantee the existence of two different pseudolines $g_1$ and $g_2$ with $\phi(g_1)=\phi(g_2)=2k$. We then reverse the modifications and change $g_1$ or $g_2$ if necessary. When we reach the original diagram, we will be able to argue that either $\phi(g_1)$ or $\phi(g_2)$ are still large enough.

More precisely, we get a sequence $\mathcal{P}_0, \mathcal{P}_1, \ldots, \mathcal{P}_h$ of extended diagrams where each $\mathcal{P}_{i+1}$ can be generated from $\mathcal{P}_i$ by one of two moves, that will be described soon. For each $\mathcal{P}_i$ we construct two different pseudolines $g_1^{(i)}$ and $g_2^{(i)}$ in such a way that allows us to argue that $\phi(g_1^{(h)})=\phi(g_1^{(h)})=2k$ and $\phi(g_1^{(0)})+\phi(g_1^{(0)})=2k+4$.

We start by describing the goal diagram $\mathcal{P}_h$. Let $r_1^{(0)}$ be the bottommost pseudoline at left infinity of the wiring diagram of $\mathcal{B^*}$ and let $r_1^{(i)}$ be the corresponding pseudoline in $\mathcal{P}_i$. Consider $c(r_1^{(h)})$, i.e.\ the color sequence of $r_1^{(h)}$ in the goal diagram $\mathcal{P}_h$. We want to have $c(r_1^{(h)})=brbrbr\ldots brb$ (note that $c(r_1^{(h)})$ has length $2k-1$). We define $g_1^{(h)}$ as a pseudoline which crosses $r_1^{(h)}$ first and then always stays in distance $\epsilon$ to it, where $\epsilon$ is an arbitrarily small positive number. Similarly we define $g_2^{(h)}$ as a pseudoline which always stays in distance $\epsilon$ to $r_1^{(h)}$ but crosses it at the end. Then $g_1^{(h)}$ and $g_2^{(h)}$ have the color sequences $c(g_1^{(h)})=rbrbrbr\ldots brb$ and $c(g_2^{(h)})=brbrbr\ldots brbr$, i.e.\ $\phi(g_1^{(h)})=\phi(g_2^{(h)})=2k$. See Figure \ref{Fig:GoalDiagram} for an illustration.

\begin{figure}
\label{Fig:GoalDiagram}
\centering
\includegraphics[scale=0.7]{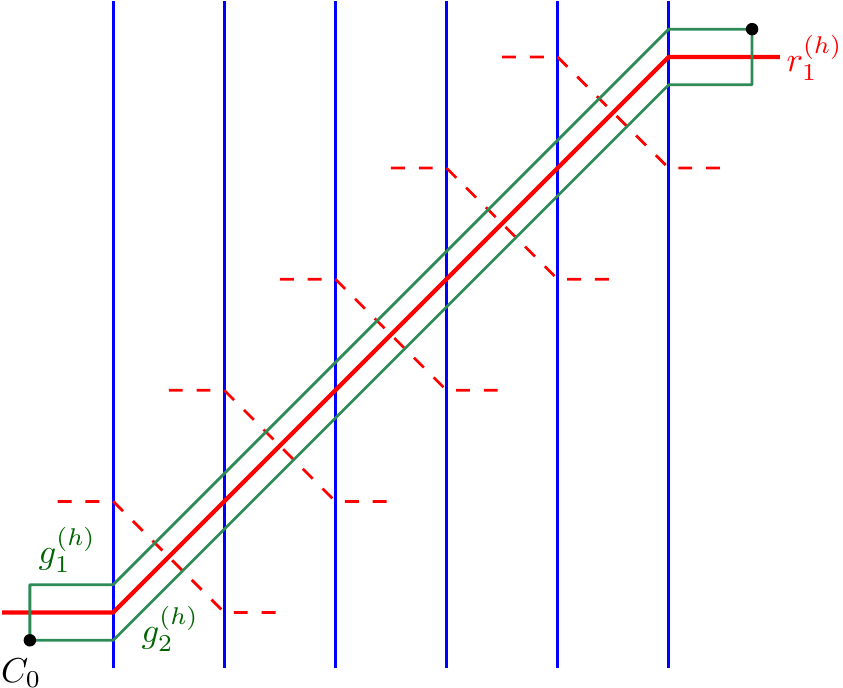}
\caption{The goal diagram $\mathcal{P}_h$ with the pseudolines $g_1^{(h)}$ and $g_2^{(h)}$.}
\end{figure}

Now we describe the moves that we allow. There are two different moves, see Figures \ref{Fig:SingleMove} and \ref{Fig:SplitMove} for illustrations of the moves.

\begin{itemize}
\item \textbf{Single crossing move:} We first describe the \emph{right} single crossing move. This move can be used when at some place in the arrangement $r_1^{(i)}$ crosses a blue pseudoline $b_1^{(i)}$, then a red pseudoline $r_2^{(i)}$ and then two blue pseudolines $b_2^{(i)}$ and $b_3^{(i)}$. Let $v_{s,t}$ denote the point of intersection of the blue pseudoline $b_s^{(i)}$ and the red pseudoline $r_t^{(i)}$. The move is as follows:
\begin{enumerate}
\item Delete the parts of $r_1^{(i)}$ between $v_{1,1}$ and $v_{2,1}$, as well as between $v_{2,1}$ and $v_{3,1}$. Delete the parts of $r_2^{(i)}$ between $v_{1,2}$ and $v_{2,2}$, as well as between $v_{2,2}$ and $v_{3,2}$.
\item Connect $v_{1,1}$ to $v_{2,2}$ and $v_{2,2}$ to $v_{3,1}$ to get a new red pseudoline $r_1^{(i+1)}$. Connect $v_{1,2}$ to $v_{2,1}$ and $v_{2,1}$ to $v_{3,2}$ to get a new red pseudoline $r_2^{(i+1)}$.
\end{enumerate}
Note that the arrangement of red pseudolines is again a wiring diagram. Also, no blue line has changed, hence the new arrangement is again an extended diagram. The color sequence of $r_1$ has changed from $c(r_1^{(i)})=\ldots brbb\ldots$ to $c(r_1^{(i+1)})=\ldots bbrb\ldots$. The \emph{left} single crossing move is defined mirrored, changing the color sequence of $r_1$ from $c(r_1^{(i)})=\ldots bbrb\ldots$ to $c(r_1^{(i+1)})=\ldots brbb\ldots$.

\item \textbf{Split crossing move:} We first describe the \emph{right} split crossing move. This move can be used when at some place in the arrangement $r_1^{(i)}$ crosses a blue pseudoline $b_1^{(i)}$, then at least two red pseudolines $r_2^{(i)},\ldots, r_m^{(i)}$ and then two blue pseudolines $b_2^{(i)}$ and $b_3^{(i)}$ somewhere in the arrangement. Let again $v_{s,t}$ denote the point of intersection of the blue line $b_s^{(i)}$ and the red pseudoline $r_t^{(i)}$. The move is as follows:
\begin{enumerate}
\item Delete the parts of $r_1^{(i)}$ between $v_{1,1}$ and $v_{2,1}$, as well as between $v_{2,1}$ and $v_{3,1}$. Delete the parts of $r_m^{(i)}$ between $v_{1,m}$ and $v_{2,m}$, as well as between $v_{2,m}$ and $v_{3,m}$.
\item Connect $v_{1,1}$ to $v_{2,m}$ and $v_{2,m}$ to $v_{3,1}$ to get a new red pseudoline $r_1^{(i+1)}$. Connect $v_{1,m}$ to $v_{2,1}$ and $v_{2,1}$ to $v_{3,m}$ to get a new red pseudoline $r_m^{(i+1)}$.
\end{enumerate}
As above, the arrangement of red pseudolines is again a wiring diagram. Again, no blue pseudoline has changed and thus the new arrangement is again an extended diagram. The color sequence of $r_1$ has changed from $c(r_1^{(i)})=\ldots rrbb\ldots$ to $c(r_1^{(i+1)})=\ldots rbrb\ldots$. The \emph{left} split crossing move is again defined mirrored, changing the color sequence of $r_1$ from $c(r_1^{(i)})=\ldots bbrr\ldots$ to $c(r_1^{(i+1)})=\ldots brbr\ldots$.
\end{itemize}

\begin{figure}
\label{Fig:SingleMove}
\centering
\includegraphics[scale=0.7]{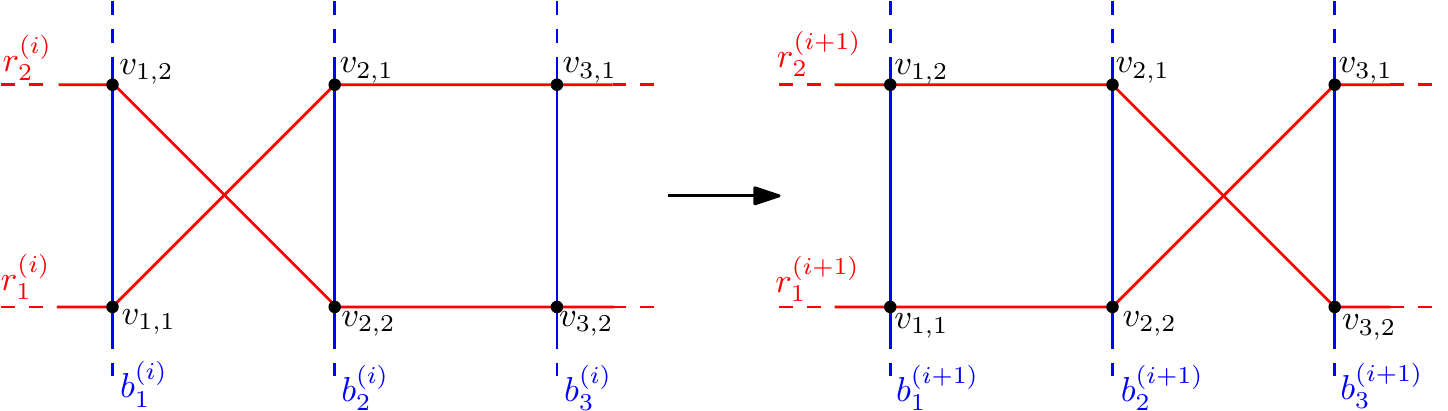}
\caption{A single crossing move.}
\end{figure}

\begin{figure}
\label{Fig:SplitMove}
\centering
\includegraphics[scale=0.7]{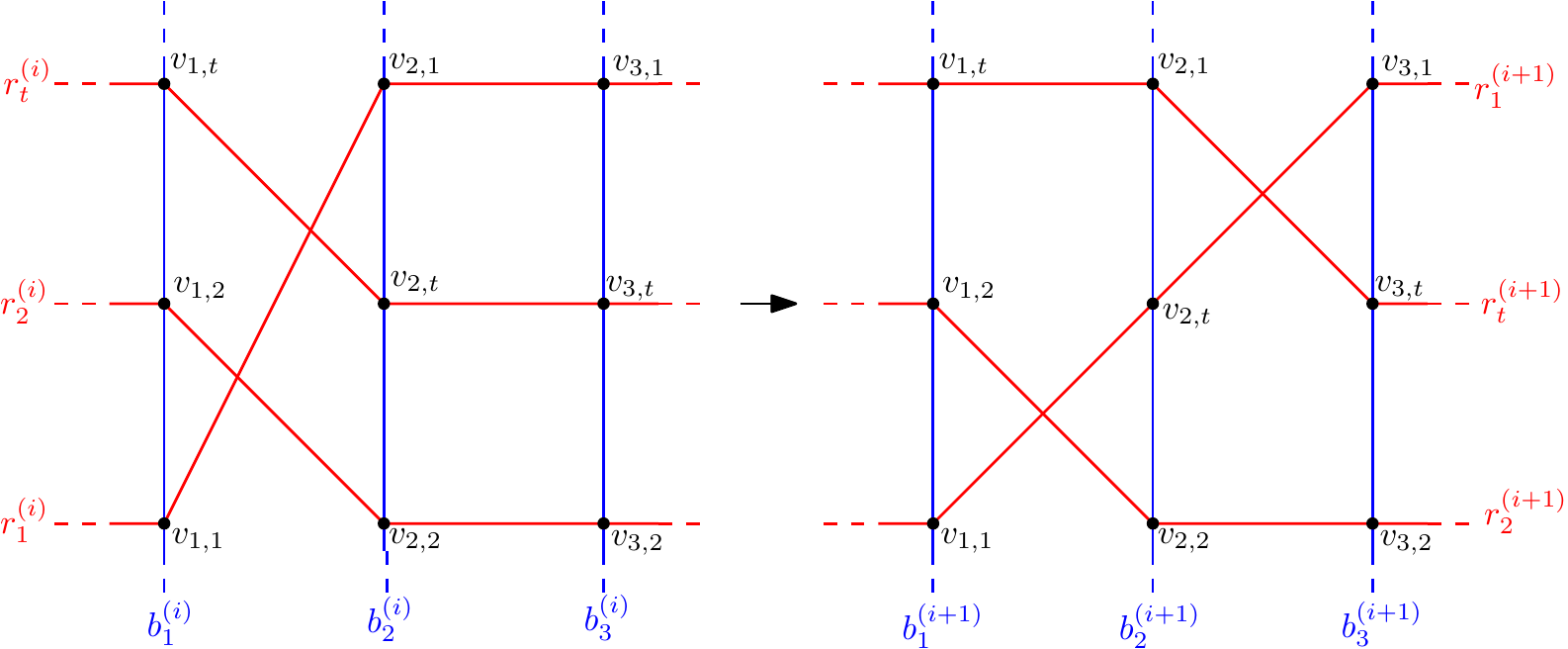}
\caption{A split crossing move.}
\end{figure}

It is easy to see that using these moves we can change the color sequence of $r_1$ from any arbitrary starting sequence with $k$ $b$'s and $k-1$ $r$'s to the sequence $c(r_1^{(h)})=brbrbr\ldots brb$. As in a split crossing move we split two consecutive $r$'s and no move joins two $r$'s, we conclude that among the moves we need to do so, there are at most $k-2$ split crossing moves. Summing up, we thus get a sequence $\mathcal{P}_0, \mathcal{P}_1, \ldots, \mathcal{P}_h$ of extended diagrams where each $\mathcal{P}_{i+1}$ can be generated from $\mathcal{P}_i$ by either a single crossing move or a split crossing move. In $\mathcal{P}_h$ we can now draw the pseudolines $g_1^{(h)}$ and $g_2^{(h)}$ as above. It remains to construct the pseudolines $g_1^{(i)}$ and $g_2^{(i)}$ for every other extended diagram $\mathcal{P}_i$. Given $g_1^{(i+1)}$ and $g_2^{(i+1)}$, we construct $g_1^{(i)}$ and $g_2^{(i)}$ as follows:

\begin{itemize}
\item \textbf{If $\mathcal{P}_{i+1}$ was generated from $\mathcal{P}_i$ by a single crossing move:} Assume without loss of generality that a right single crossing move was used (otherwise mirror the construction). As $r_1^{(i)}$ is never below $r_1^{(i+1)}$, we can set $g_2^{(i)}:=g_2^{(i+1)}$. If $g_1^{(i+1)}$ crosses $r_1^{(i)}$ only once, we also set $g_1^{(i)}:=g_1^{(i+1)}$. If however $g_1^{(i+1)}$ crosses $r_1^{(i)}$ three times, let $v_1$ be a point on $g_1^{(i+1)}$ slightly before the second crossing and let $v_2$ be a point on $g_1^{(i+1)}$ slightly after the third crossing. Delete the part of $g_1^{(i+1)}$ that lies between $v_1$ and $v_2$ and connect the two points by a pseudoline-segment that always lies at distance $\epsilon$ above $r_1^{(i)}$. The resulting pseudoline is $g_1^{(i)}$.

Note that $c(g_2^{(i)})=c(g_2^{(i+1)})$ and thus $\phi(g_2^{(i)})=\phi(g_2^{(i+1)})$. For $g_1$, the color sequence might change from $c(g_1^{(i+1)})=\ldots bbrb \ldots$ to $c(g_1^{(i)})=\ldots brbb \ldots$. However, this neither increases nor decreases the length of the longest semialternating subsequence, and so $\phi(g_1^{(i)})=\phi(g_1^{(i+1)})$. See Figure \ref{Fig:SingleMoveReverse} for an illustration.

\item \textbf{If $\mathcal{P}_{i+1}$ was generated from $\mathcal{P}_i$ by a split crossing move:} Assume again without loss of generality that a right split crossing move was used (otherwise mirror the construction). Also in this case $r_1^{(i)}$ is never below $r_1^{(i+1)}$, and so we set $g_2^{(i)}:=g_2^{(i+1)}$. Like before, if $g_1^{(i+1)}$ crosses $r_1^{(i)}$ only once, we also set $g_1^{(i)}:=g_1^{(i+1)}$. If however $g_1^{(i+1)}$ crosses $r_1^{(i)}$ three times, again let $v_1$ be a point on $g_1^{(i+1)}$ slightly before the second crossing and let $v_2$ be a point on $g_1^{(i+1)}$ slightly after the third crossing. As we did for the single crossing move, delete the part of $g_1^{(i+1)}$ that lies between $v_1$ and $v_2$ and connect the two points by a pseudoline-segment that always lies at distance $\epsilon$ above $r_1^{(i)}$. The resulting pseudoline is $g_1^{(i)}$.

Note that also in this case we have $c(g_2^{(i)})=c(g_2^{(i+1)})$ and thus $\phi(g_2^{(i)})=\phi(g_2^{(i+1)})$. For $g_1$, the color sequence changes from $c(g_1^{(i+1)})=\ldots rbrb \ldots$ to $c(g_1^{(i)})=\ldots rrbb \ldots$. Thus the longest semialternating subsequence of $c(g_1^{(i)})$ might contain one $rb$-term less than the one of $c(g_1^{(i+1)})$, but we can still ensure that $\phi(g_1^{(i)})\geq\phi(g_1^{(i+1)})-2$. See Figure \ref{Fig:SplitMoveReverse} for an illustration.
\end{itemize}

\begin{figure}
\label{Fig:SingleMoveReverse}
\centering
\includegraphics[scale=0.7]{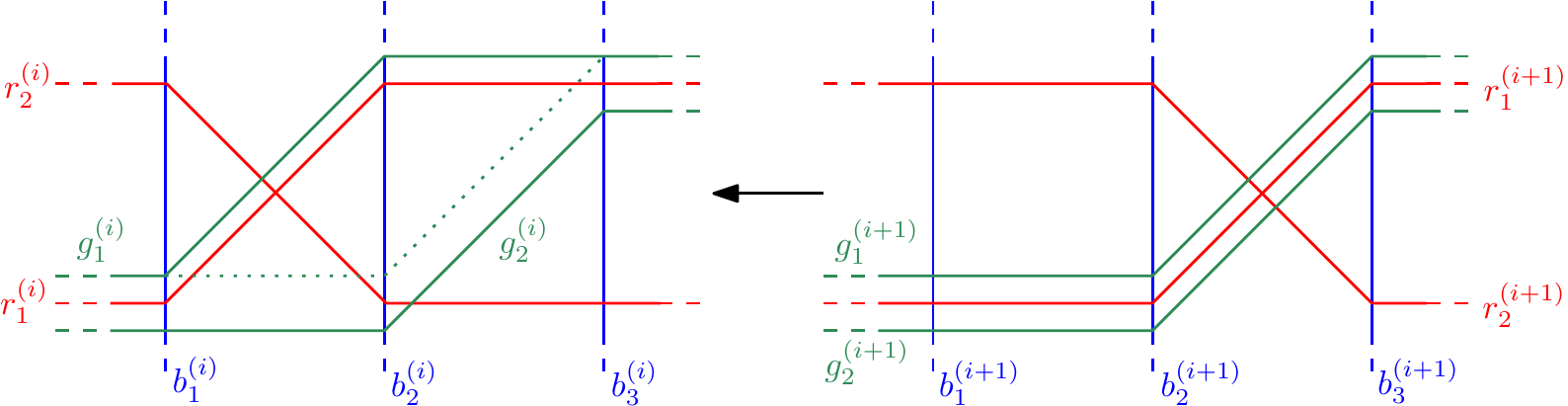}
\caption{Reversing a single crossing move.}
\end{figure}

\begin{figure}
\label{Fig:SplitMoveReverse}
\centering
\includegraphics[scale=0.7]{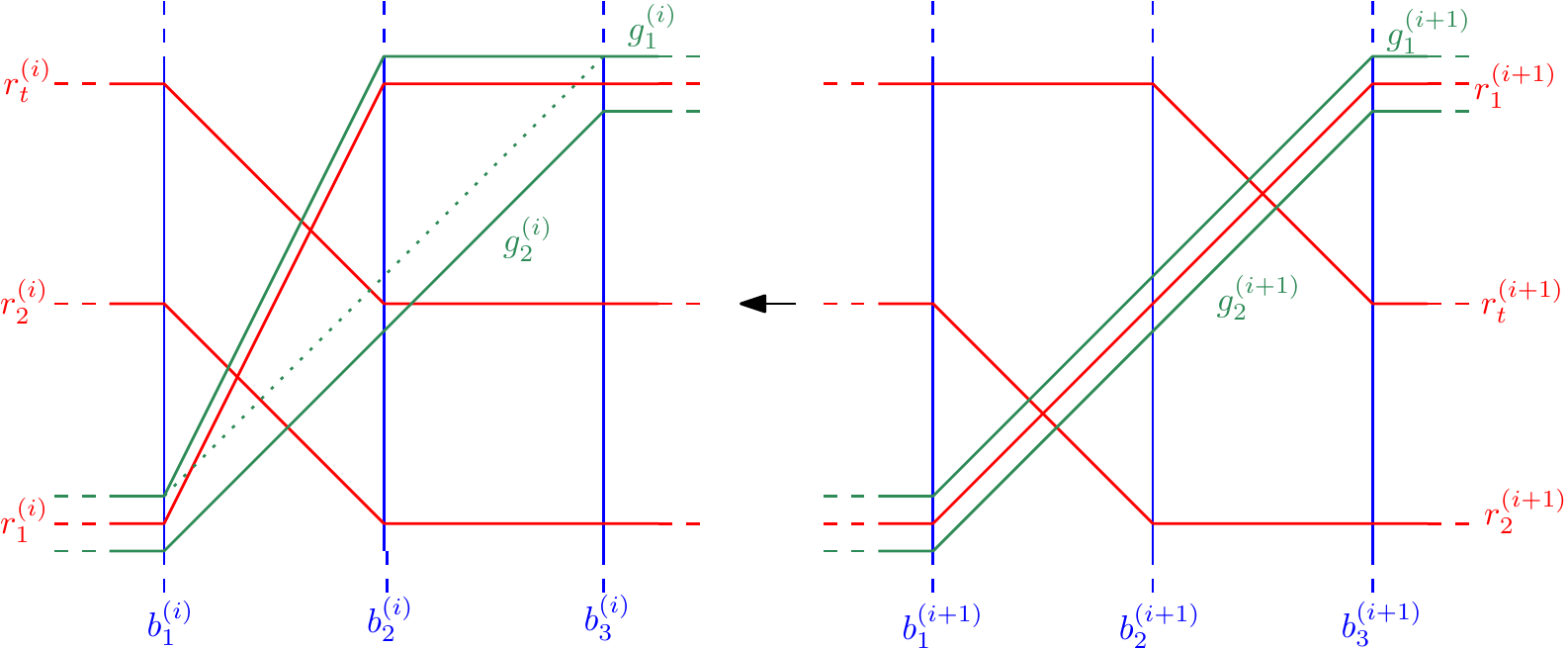}
\caption{Reversing a split crossing move.}
\end{figure}

To conclude the proof, let $\Phi(i):=\phi(g_1^{(i)})+\phi(g_2^{(i)})$. From the construction above we can see that $\Phi(i)=\Phi(i+1)-2$ only if $\mathcal{P}_{i+1}$ was generated from $\mathcal{P}_i$ by a split crossing move. Otherwise we have $\Phi(i)=\Phi(i+1)$. As we have $\Phi(h)=4k$ and we used no more than $k-2$ split crossing moves, we conclude that $\Phi(0)\geq 4k-(k-2)\cdot2=2k+4$. Then by the pigeonhole principle, remembering that $\phi(g)$ is always an even number by definition, for $j=1$ or $j=2$ we have $\phi(g_i^{(0)})\geq k+2$, if $k$ is even, or $\phi(g_i^{(0)})\geq k+3$, if $k$ is odd.
\end{proof}

\begin{cor}
\label{Cor:avoidspoke}
If a point set $\mathcal{P}$ contains two subsets $\mathcal{A}$ and $\mathcal{B}$ of size $k$, such that $\mathcal{A}$ avoids $\mathcal{B}$ and $\mathcal{A}$ and $\mathcal{B}$ can be separated by a line, then $\mathcal{P}$ contains a spoke set of size $\lceil\frac{k}{2}\rceil+1$.
\end{cor}

\begin{proof}
Combine Lemma \ref{Lem:dual} and Lemma \ref{Lem:wiring}
\end{proof}

Modifying the proof of Aronov et al.\ \cite{Aronov} for finding mutually avoiding sets in a point set, we can prove the following theorem:

\begin{theorem}
\label{Thm:avoid}
Every point set of size $n$ contains two point sets $\mathcal{A}$ and $\mathcal{B}$ of size $\lfloor\sqrt{\frac{n}{2}+1}-1\rfloor$ such that $\mathcal{A}$ avoids $\mathcal{B}$ and $\mathcal{A}$ and $\mathcal{B}$ can be separated by a line.
\end{theorem}

\begin{proof}
Let $t:=\sqrt{\frac{n}{2}+1}-1$. Move a vertical line $\ell_1$ from $x=-\infty$ to the right until $\lfloor 2t\rfloor$ points are to the left of it. Call the set of these $\lfloor 2t\rfloor$ points $\mathcal{L}$. Do the same from the right, that is, find a vertical line $\ell_2$ that has $\lfloor 2t\rfloor$ points to the right of it. Call the set of these $\lfloor 2t\rfloor$ points $\mathcal{R}$. Denote the set of the remaining points, i.e.\ the points that lie between $\ell_1$ and $\ell_2$ by $\mathcal{M}$. Note that $\mathcal{M}$ has size $n-2\lfloor 2t\rfloor$. Perform a Ham-Sandwich-cut on $\mathcal{L}$ and $\mathcal{R}$. The resulting line $\ell_3$ splits any point set $\mathcal{S}$, $\mathcal{S}\in\{\mathcal{L},\mathcal{R},\mathcal{M}\}$, into $\mathcal{S}_a$ and $\mathcal{S}_b$, where $\mathcal{S}_a$ lies above $\ell_3$ and $\mathcal{S}_b$ lies below $\ell_3$. Note that $|\mathcal{L}_a|$, $|\mathcal{L}_b|$, $|\mathcal{R}_a|$ and $|\mathcal{R}_b|$ all have size at least $\lfloor t\rfloor$. Assume without loss of generality that $|\mathcal{M}_a|\geq |\mathcal{M}_b|$. In particular, $|\mathcal{M}_a|\geq\lceil\frac{n}{2}-2t\rceil$. Apply an affine transformation such that $\ell_1$ and $\ell_2$ remain vertical, and $\ell_3$ becomes horizontal. Order the set $\mathcal{M}_a$ from left to right. A well known result states that any sequence of real numbers of length $r$ contains either an ascending or a descending subsequence of length $\sqrt{r}$. Apply this result to the $y$-coordinates of $\mathcal{M}_a$ to find a subset $\mathcal{A}$ of $\mathcal{M}_a$ with $|\mathcal{A}|\geq\sqrt{\lceil\frac{n-4t}{2}\rceil}\geq\sqrt{\frac{n-4t}{2}}$ that is without loss of generality descending. Note that $\mathcal{A}$ avoids $\mathcal{L}_b$ and that $\ell_3$ separates these two sets. Setting $\mathcal{B}:=\mathcal{L}_b$ and observing that $\sqrt{\frac{n-4t}{2}}=t$ for $t=\sqrt{\frac{n}{2}+1}-1$ finishes the proof.
\end{proof}

\begin{cor}
Every point set of size $n$ allows a spoke set of size at least $\sqrt{\frac{n}{8}}$.
\end{cor}

\begin{proof}
By Theorem \ref{Thm:avoid}, there are two sets $\mathcal{A}$ and $\mathcal{B}$ of size $\lfloor\sqrt{\frac{n}{2}+1}-1\rfloor$ such that $\mathcal{A}$ avoids $\mathcal{B}$ and $\mathcal{A}$ and $\mathcal{B}$ can be separated by a line. Thus, by Corollary \ref{Cor:avoidspoke}, the point set contains a spoke set of size
$$\left\lceil\frac{\lfloor\sqrt{\frac{n}{2}+1}-1\rfloor}{2}\right\rceil+1 \geq \left\lceil\sqrt{\frac{n}{8}+\frac{1}{4}}-1\right\rceil+1 \geq \sqrt{\frac{n}{8}}.$$
\end{proof}

It is worth mentioning that there are point sets that have no mutually avoiding subsets of size larger than $\mathcal{O}(\sqrt{n})$ \cite{Valtr}. However, it is not clear whether this still holds if we only insist that one of the subsets avoids the other one. So while there is no hope of finding larger crossing families by finding larger mutually avoiding subsets, it might still be possible to find larger spoke sets with this approach.

\section{Spoke matchings}
\label{Sec:matchings}
In this section we characterize a family of geometric matchings that arise from spoke sets. Consider a point set $\mathcal{P}$ that has a spoke set of size $k$. Let $\mathcal{Q}$ be a subset of $\mathcal{P}$ that has exactly one point in each unbounded region. Recall that the \emph{spoke matching} of $\mathcal{Q}$ is the matching obtained by connecting with a straight line segment each point $p$ in $\mathcal{Q}$ to the unique point $q$ in $\mathcal{Q}$ that lies in the antipodal unbounded region of the spoke set. In order to characterize the geometric properties of spoke matchings, we need a few definitions:

Let $e$ and $f$ be two edges and let $s$ be the intersection of their supporting lines. If $s$ lies in both $e$ and $f$, we say that $e$ and $f$ \emph{cross}. If $s$ lies in $f$ but not in $e$, we say that $e$ \emph{stabs} $f$ and we call the vertex of $e$ that is closer to $s$ the \emph{stabbing vertex} of $e$. If $s$ lies neither in $e$ nor in $f$, or even at infinity, we say that $e$ and $f$ are \emph{parallel}.

A \emph{stabbing chain} are three edges, $e$, $f$ and $g$, where $e$ stabs $f$ and $f$ stabs $g$. We call $f$ the \emph{middle} edge of the stabbing chain.

\begin{theorem}
\label{Thm:charac}
A geometric matching $M$ is a spoke matching if and only if it satisfies the following three conditions:
\begin{description}
\item[(a)] no two edges are parallel,
\item[(b)] if an edge $e$ stabs two other edges $f$ and $g$, then the respective stabbing vertices of $e$ lie inside the convex hull of $f$ and $g$, and
\item[(c)] if there is a stabbing chain, then the stabbing vertex of the middle edge lies inside the convex hull of the other two edges.
\end{description}
\end{theorem}

Note that the fact that every crossing family of size $k$ induces a spoke set of size $k$ can also be derived from this result, as it shows that every crossing family is a spoke matching. However, the family of spoke matchings also matchings that are not crossing families. In fact, it is even possible to construct a crossing-free spoke matching.
In \cite{MT}, it has been shown that there are sets of $n$ points in general position that do not allow any matching satisfying conditions (a), (b) and (c) of size larger than $\frac{9}{20}n$. Hence we get the following corollary:

\begin{cor}
There are point sets $\mathcal{P}$ of $n$ points in general position that do not admit a spoke set of size larger than $\frac{9}{20}n$.
\end{cor}

We conclude with the proof of Theorem \ref{Thm:charac}.

\begin{proof}
Note that if the matching consists of two edges, it is easy to check that the statement is true. So in the following, we will assume that $M$ has at least three edges.
We start by showing that every spoke matching satisfies the three conditions:

Let $\mathcal{S}$ be the spoke set that gives rise to the spoke matching $M$. We start with condition (a). Pick two edges $e=(p,q)$ and $f=(r,s)$ from $M$ and two lines $\ell_1$ and $\ell_2$ from $\mathcal{S}$, with the property that one of the endpoints of $e$ and $f$ lies in each unbounded region defined by $\ell_1$ and $\ell_2$. Assume without loss of generality that $\ell_1$ is horizontal and $\ell_2$ is vertical. Let $p$ be in the bottom left region and let $r$ be in the bottom right region. Then $q$ is in the top right region and $s$ is in the top left region. Assume without loss of generality that $e$ intersects the top left region and let $T_1$ be the triangle bounded by $e$, $\ell_1$ and $\ell_2$. If $s$ lies in $T_1$, then $f$ stabs $e$. Assume without loss of generality that $f$ intersects the top right region and let $T_2$ be the triangle bounded by $f$, $\ell_1$ and $\ell_2$. If $q$ lies in $T_2$, then $e$ stabs $f$. If $s$ does not lie in $T_1$ and $q$ does not lie in $T_2$, then $e$ and $f$ cross. Thus any two edges in $M$ are either crossing or stabbing, and thus $M$ satisfies condition (a).

Now we show that $M$ satisfies conditions (b) and (c) by proving that for three edges $e$, $f$ and $g$, with $f$ stabbing $e$, the stabbing vertex of $f$ lies in the convex hull of $e$ and $g$. So, pick three edges $e=(p,q)$, $f=(r,s)$ and $g=(t,u)$ from $M$ and three lines $\ell_1$, $\ell_2$ and $\ell_3$ from $\mathcal{S}$, with the property that one of the endpoints of $e$, $f$ and $g$ lies in each unbounded region defined by $\ell_1$, $\ell_2$ and $\ell_3$. Assume without loss of generality that $f$ stabs $e$ with stabbing vertex $s$. Let $A_1,\ldots,A_6$ be the unbounded regions defined by $\ell_1$, $\ell_2$ and $\ell_3$, and assume without loss of generality that $s\in A_1$, $p\in A_2$, $t\in A_3$, $r\in A_4$, $q\in A_5$ and $u\in A_6$. Let $\ell_1$ be the line separating $A_1$, $A_2$ and $A_3$ from $A_4$, $A_5$ and $A_6$. Let $\ell_2$ be the line separating $A_2$, $A_3$ and $A_4$ from $A_1$, $A_5$ and $A_6$. Finally, let $\ell_3$ be the line separating $A_3$, $A_4$ and $A_5$ from $A_1$, $A_2$ and $A_6$.

As $f$ stabs $e$ with stabbing vertex $s$, the edge $e$ must intersect $A_1$. Let $R$ be the part of $A_1$ that is bounded by $e$ and note that $s$ lies in $R$. Let $H$ be the convex hull of $e$ and $g$. We will show that $R\subset H$. Consider the edge $(q,t)$. This edge does not intersect the line $\ell_3$ as $q\in A_5$ and $t\in A_3$. Similarly, the edge $(p,t)$ does not intersect the line $\ell_2$. Let $T$ be the triangle defined by the edges $(q,t)$, $(p,t)$ and $e=(p,q)$. As $(q,t)$ does not cross $\ell_3$ and $(p,t)$ does not cross $\ell_2$, we deduce that $R\subset T$. Clearly $T\subset H$. Thus we see that $R\subset H$ and as $s$ lies in $R$, $s$ also lies in $H$. Thus for any three edges $e$, $f$ and $g$, with $f$ stabbing $e$, the stabbing vertex of $f$ lies inside the convex hull of $e$ and $g$. This proves that $M$ satisfies the conditions (b) and (c). See Figure \ref{Fig:BoseSuff} for an illustration of this argument.

\begin{figure}
\centering
\includegraphics[scale=1]{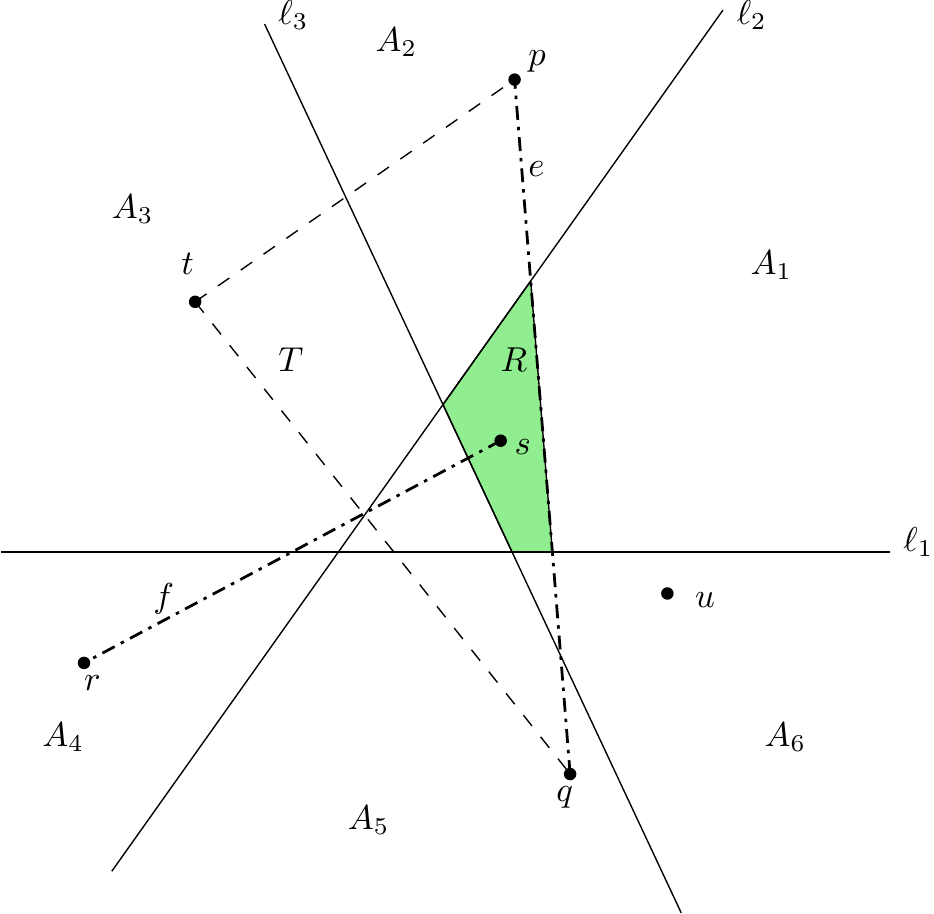}
\caption{An illustration of the proof that every spoke matching satisfies the conditions (b) and (c).}
\label{Fig:BoseSuff}
\end{figure}

It remains to show that any matching that satisfies conditions (a), (b) and (c) is indeed a spoke matching, i.e.\ there is a set $\mathcal{S}$ of pairwise non-parallel lines such that in every unbounded region there is an endpoint of $M$ and each edge of $M$ intersects all lines of $\mathcal{S}$. Let $m_1, \ldots, m_k$ be the edges of the matching, ordered by slope. Assume without loss of generality that no edge has slope $0$. We show that for any $i\leq k$, we can find a line $\ell_i$ whose slope is between the slopes of $m_i$ and $m_{i+1}$ (where $m_{k+1}=m_1$) and which intersects all edges $m_1, \ldots, m_k$. Without loss of generality, let $m_i$ be the last edge that has negative slope and let $m_{i+1}$ be the first edge that has positive slope. For each edge in $M$, color the upper endpoint red and the lower endpoint blue. Let $\mathcal{R}$ be the set of red points and let $\mathcal{B}$ be the set of blue points. We want to show that we can separate the blue points and the red points by a line or, equivalently, that the convex hulls of $\mathcal{R}$ and $\mathcal{B}$ are disjoint.

Assume for the sake of contradiction that there is a red point $r_1$ inside the convex hull of $\mathcal{B}$. Then, there must be two blue points $b_2$ and $b_3$, such that $b_2$ is to the left of $r_1$, $b_3$ is to the right of $r_1$, and $r_1$ lies below the line through $b_2$ and $b_3$. Let $b_1$ be the blue point that is matched with $r_1$ in $M$, and let $r_2$ and $r_3$ be the red points that are matched with $b_2$ and $b_3$, respectively. Assume first that $b_1\neq b_2$ and $b_1\neq b_3$. For $j\in\{1,2,3\}$, let $e_j$ be the edge between $b_j$ and $r_j$. Note that as $b_j$ lies below $r_j$ for every $j\in\{1,2,3\}$, the line through $b_2$ and $b_3$ separates $e_1$ from $e_2$ and $e_3$. Thus $e_1$ cannot cross any of these edges. If $e_1$ is parallel to $e_2$ or $e_3$, we have a contradiction to condition (a). If $e_1$ stabs both $e_2$ and $e_3$, we have a contradiction to condition (b). If $e_1$ stabs $e_2$ and gets stabbed by $e_3$, or vice versa, we have a contradiction to condition (c). So assume that $e_1$ gets stabbed by both $e_2$ and $e_3$. Let $s_2$ and $s_3$ be the supporting lines of $e_2$ and $e_3$, respectively. As $e_2$ stabs $e_1$, by definition $s_2$ intersects $e_1$. This implies that $r_1$ must lie above $s_2$. By the same argument, $r_1$ must also lie above $s_3$. In particular, the intersection of $s_2$ and $s_3$ lies below $r_1$ and therefore below the line through $b_2$ and $b_3$. Hence $s_2$ does not intersect $e_3$ and $s_3$ does not intersect $e_2$. But this implies that $e_2$ and $e_3$ are parallel, which is a contradiction to condition (a). See Figure \ref{Fig:BoseNec1} for an illustration of this argument.

\begin{figure}
\centering
\includegraphics[scale=0.7]{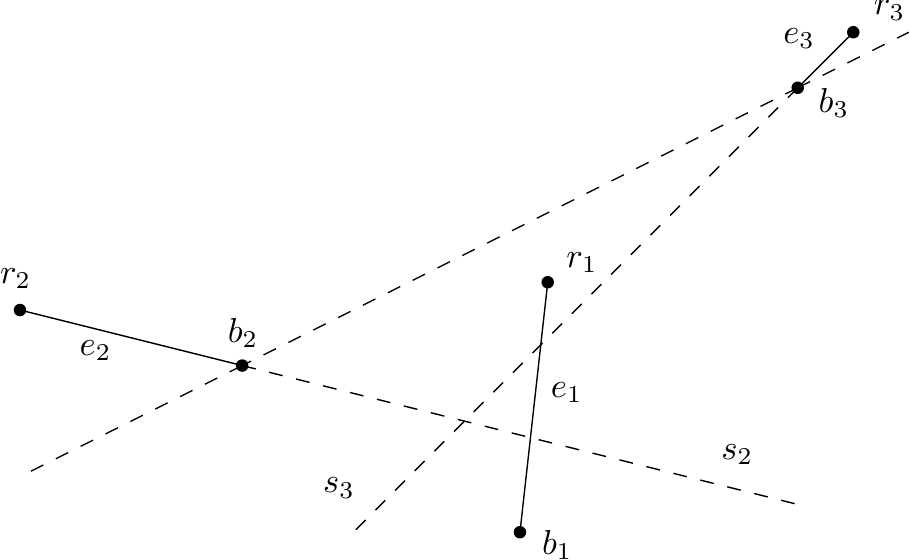}
\caption{An illustration of the proof that every matching that satisfies the conditions (a),(b) and (c) is a spoke matching.}
\label{Fig:BoseNec1}
\end{figure}

Assume now without loss of generality that $b_1=b_2$. We can further assume that $b_2$ is the only point of $\mathcal{B}$ to the left of $r_1$, as otherwise we could just choose another point of $\mathcal{B}$ as $b_2$. Let $b_4$ be another blue point such that $r_1$ lies inside the convex hull of $b_2$, $b_3$ and $b_4$. Let $r_3$ and $r_4$ be the red points that are matched with $b_3$ and $b_4$, respectively, and let $e_3$ and $e_4$ be the respective segments. Denote the segment between $r_1$ and $b_2$ by $e_1$. Again, let $s_i$ denote the supporting line of $e_i$. As $r_1$ is matched with $b_2$, $b_2$ has to lie below $r_1$. As $r_1$ lies below the line through $b_2$ and $b_3$, this implies that $b_3$ lies above $r_1$ and the line through $b_2$ and $b_3$ separates $e_1$ and $e_3$. In particular, $e_1$ and $e_3$ cannot cross and $e_1$ cannot stab $e_3$. As $e_1$ and $e_3$ being parallel would be a contradiction to condition (a), we assume that $e_3$ stabs $e_1$. Note that in this case $b_3$ has to be the stabbing vertex and thus $r_3$ has to be to the right of $b_3$. Now consider $e_4$. If $e_4$ is parallel to $e_3$ or $e_1$, we have a contradiction to condition (a). If $e_4$ stabs $e_1$, then $s_4$ intersects $e_1$. In particular, $e_4$ cannot cross or stab $e_3$, thus it must get stabbed by $e_3$. But $b_3$ cannot be in the convex hull of $e_1$ and $e_3$, so we have a contradiction to condition (b). By a similar argument we get a contradiction to condition (c) if $e_4$ stabs $e_3$. So assume that $e_4$ gets stabbed by both $e_1$ and $e_3$. But then $s_4$ separates $e_1$ and $e_3$, and again $b_3$ is not in the convex hull of $e_1$ and $e_4$, which is again a contradiction to condition (a).

Thus, by contradiction, the convex hull of $\mathcal{R}$ and the convex hull of $\mathcal{B}$ are disjoint, i.e.\ they can be separated by a line $\ell_i$. By our choice of $\mathcal{R}$ and $\mathcal{B}$, this line intersects all edges of $M$ and its slope is between the slopes of $m_i$ and $m_{i+1}$.

Finally, $\mathcal{S}:=\{\ell_1, \ldots, \ell_k\}$ is indeed a spoke set: The condition that each line of $\mathcal{S}$ crosses each edge of $M$ ensures that each endpoint of an edge in $M$ lies in an unbounded region of the arrangement determined by $\mathcal{S}$. Also, as each $\ell_i$ has a slope between the slopes of $m_i$ and $m_{i+1}$, no two endpoints are in the same region of the arrangement. As $|\mathcal{S}|=|M|$, there has to be exactly one endpoint in each unbounded region.
\end{proof}

\section*{Conclusion}
We have shown that every set of $n$ points in general position contains a spoke set of size $\sqrt{\frac{n}{8}}$. This bound is better than the one immediately derived from crossing families. For some constructions where crossing families are used, spoke sets are actually enough, so this result also improves some other bounds, e.g.\ the number of crossing-free spanning trees needed to cover the edge set of a complete geometric graph \cite{Bose}.

A question that remains open is of course whether it is always possible to find a spoke set of linear size. Using our approach, it would be enough to find two subsets $\mathcal{A}$ and $\mathcal{B}$ of linear size, that can be separated by a line and such that $\mathcal{A}$ avoids $\mathcal{B}$. As mentioned above, it is known that this is not always possible for mutually avoiding sets, but it is open whether it is possible for the more general case that we need. Alternatively, one could try to construct long spoke paths for arbitrary line arrangements.

Another interesting question is whether it is always possible to find a crossing family of size linear in the size of the largest spoke set. Theorem \ref{Thm:charac} is a first step in this direction as it shows that even though spoke matchings are different from crossing families (and could in fact even be crossing-free), they still satisfy a number of conditions.

\label{Sec:conclusion}

\bibliography{refs.bib}

\begin{thebibliography}{1}

\bibitem{Aronov}
Boris Aronov, Paul Erd{\"{o}}s, Wayne Goddard, Daniel~J. Kleitman, Michael
  Klugerman, J{\'{a}}nos Pach, and Leonard~J. Schulman.
\newblock Crossing families.
\newblock In {\em Proceedings of the Seventh Annual Symposium on Computational
  Geometry, North Conway, NH, USA, June 10-12, 1991}, pages 351--356, 1991.

\bibitem{Bose}
Prosenjit Bose, Ferran Hurtado, Eduardo Rivera-Campo, and David~R. Wood.
\newblock Partitions of complete geometric graphs into plane trees.
\newblock {\em Computational Geometry}, 34(2):116 -- 125, 2006.

\bibitem{Fulek}
Radoslav Fulek and Andrew Suk.
\newblock On disjoint crossing families in geometric graphs.
\newblock {\em Electronic Notes in Discrete Mathematics}, 38:367--375, 2011.

\bibitem{Pach}
J{\'{a}}nos Pach and J{\'{o}}zsef Solymosi.
\newblock Halving lines and perfect cross-matchings.
\newblock {\em Advances in Discrete and Computational Geometry}, 223:245--249,
  1999.

\bibitem{MT}
Patrick Schnider.
\newblock Partitions and packings of complete geometric graphs with plane
  spanning double stars and paths.
\newblock Master's thesis, ETH Z{\"{u}}rich, 2015.

\bibitem{Valtr}
Pavel Valtr.
\newblock On mutually avoiding sets.
\newblock In {\em The mathematics of {P}aul Paul {E}rd{\"{o}}s, {II (R. L.
  Graham and J. Nesetril, eds.) Algorithms and Combin. 14}}, pages 324--332,
  1997.

\end{thebibliography}
\bibliographystyle{plain}

\end{document}